\documentclass[a4paper,fleqn,11pt]{article}

\usepackage{amsmath}
\usepackage{amsthm}
\usepackage{amssymb}
\usepackage{accents}

\usepackage[a4paper,top=3cm, bottom=3cm, left=3cm, right=3cm]{geometry}
\usepackage[shortlabels]{enumitem}
\usepackage[pdftex]{graphicx}
\usepackage{subcaption}
\usepackage[pdftex,ocgcolorlinks,pagebackref=false]{hyperref}
\usepackage[affil-it]{authblk}

\theoremstyle{plain}
\newtheorem{theorem}{Theorem}[section]
\newtheorem{lemma}[theorem]{Lemma}
\newtheorem{proposition}[theorem]{Proposition}
\newtheorem{corollary}[theorem]{Corollary}
\newtheorem{remark}[theorem]{Remark}

\theoremstyle{definition}
\newtheorem{definition}[theorem]{Definition}

\newcommand{\asymptoticle}{\lesssim} 
\newcommand{\asymptoticge}{\gtrsim} 
\newcommand{\norm}[2][]{\left\|#2\right\|_{#1}}
\newcommand{\ket}[1]{\left|#1\right\rangle}
\newcommand{\bra}[1]{\left\langle #1\right|}
\newcommand{\ketbra}[2]{\left|#1\middle\rangle\!\middle\langle#2\right|}

\newcommand{\setbuild}[2]{\left\{#1\middle|#2\right\}}
\DeclareMathOperator{\boundeds}{\mathcal{B}}
\DeclareMathOperator{\Tr}{Tr}
\DeclareMathOperator{\rank}{rk}
\DeclareMathOperator{\support}{supp}
\DeclareMathOperator{\spectrum}{spec}
\DeclareMathOperator{\id}{id}
\newcommand{\entropy}{H}
\newcommand{\relativeentropy}[3][]{\mathop{D_{#1}}\mathopen{}\left(#2\middle\|#3\right)\mathclose{}}

\newcommand{\preorderle}{\preccurlyeq}
\newcommand{\preorderge}{\succcurlyeq}
\DeclareMathOperator{\ev}{ev}
\newcommand{\sandwicheddivergence}[3]{\widetilde{D}_{#1}(#2\|#3)}
\newcommand{\thesandwichedquasientropy}[1]{\widetilde{Q}_{#1}}
\newcommand{\sandwichedquasientropy}[3]{\thesandwichedquasientropy{#1}(#2\|#3)}

\newcommand{\dichotomies}{\mathcal{D}}
\newcommand{\cdichotomies}{\mathcal{D}_\textnormal{c}}
\newcommand{\pinching}[1]{\mathcal{P}_{#1}}

\title{The semiring of dichotomies and asymptotic relative submajorization}
\author[1]{Christopher Perry}
\author[2,3]{P\'eter Vrana}
\author[1,4]{Albert~H. Werner}
\affil[1]{QMATH, Department of Mathematical Sciences, University of Copenhagen, Universitetsparken~5, 2100~Copenhagen, Denmark}
\affil[2]{Institute of Mathematics, Budapest University of Technology and Economics, Egry~J\'ozsef u.~1., Budapest, 1111 Hungary}
\affil[3]{MTA-BME Lend\"ulet Quantum Information Theory Research Group}
\affil[4]{NBIA, Niels Bohr Institute, University of Copenhagen, Blegdamsvej~17, 2100~Copenhagen, Denmark}

\begin{document}
\maketitle

\begin{abstract}
We study quantum dichotomies and the resource theory of asymmetric distinguishability using a generalization of Strassen's theorem on preordered semirings. We find that an asymptotic variant of relative submajorization, defined on unnormalized dichotomies, is characterized by real-valued monotones that are multiplicative under the tensor product and additive under the direct sum. These strong constraints allow us to classify and explicitly describe all such monotones, leading to a rate formula expressed as an optimization involving sandwiched R\'enyi divergences. As an application we give a new derivation of the strong converse error exponent in quantum hypothesis testing.
\end{abstract}

\section{Introduction}

The resource theoretic study of hypothesis testing was initated by Matsumoto in \cite{matsumoto2010reverse}. In this resource theory the objects we are comparing are pairs of quantum states -- also called dichotomies -- on a Hilbert space, representing two alternative states a system might be in, and the allowed operations are arbitrary quantum channels, simultaneously applied to both states. From this point of view, hypothesis testing is about distilling (i.e. transforming a given pair into) standard pairs, but it is meaningful to consider more general transformations between arbitrary dichotomies, including the reverse process. Building upon this information-theoretic approach, recent work \cite{buscemi2019information,wang2019resource} found the optimal rates of asymptotic pair transformations, when one demands exact transformation of the second component and only approximate of the first one, with asymptotically vanishing error. There it is also found that the convergence of the error to $0$ and $1$ on either side of this rate is exponentially fast.

When studying resource theories with a tensor-product structure (e.g. in the sense of \cite[Definition 2.]{chitambar2019quantum}), one is often interested in how composition of resources interacts with the preorder induced by the free operations. This structure allows one to define rates of exact asymptotic transformations by comparing large tensor powers of the initial and final states (or more general objects). Under fairly general conditions such rates are characterized by additive real-valued monotones \cite{fritz2017resource}. Interestingly, a result of this type has appeared much earlier, in Strassen's work on the asymptotic comparison of large powers tensors \cite{strassen1988asymptotic}. In that case there is an additional operation (direct sum) and, remarkably, one need only consider monotones that respect both operations, i.e. monotone semiring homomorphisms into the nonnegative reals. The collection of such monotone homomorphisms is called the asymptotic spectrum of the preordered semiring, and assuming a boundedness condition they characterize an asymptotic relaxation of the preorder. Recent works have shown that Strassen's theory of asymptotic spectra can be applied to resource theories relevant to classical and quantum information theory, leading to powerful characterizations of various zero-error capacities \cite{zuiddam2019asymptotic,li2018quantum} as well as strong converse exponents for entanglement transformations via local operations and classical communication \cite{jensen2019asymptotic}.

In this paper we employ the method of asymptotic spectra to study quantum dichotomies. We start with the observation that it is possible to turn the set of (equivalence classes of) unnormalized dichotomies into a semiring in such a way that the preorder given by relative submajorization \cite{renes2016relative} is compatible with the direct sum and the tensor product. While the boundedness condition of Strassen's theorem is not satisfied in the resulting preordered semiring, a recent generalization of that theorem \cite{vrana2020generalization} does apply. The asymptotic preorder captures probabilistic asymptotic transformations in the strong converse regime, and encodes the strong converse exponents for pair transformations (with a restricted type of approximation in the first component, as allowed by relative submajorization). The generalization of Strassen's theorem allows one to characterize the asymptotic preorder in terms of the asymptotic spectrum, but does so in a non-constructive way. Nevertheless, we are able to determine explicitly the set of monotone semiring homomorphisms on our preordered semiring, identifying them as sandwiched R\'enyi quasi-entropies \cite{muller2013quantum,wilde2014strong} of orders $\alpha\ge 1$. This result can be directly translated into explicit formulas for the strong converse exponents involving an optimization over a single parameter $\alpha$. The strong converse exponent for hypothesis testing \cite{mosonyi2015quantum} emerges as a special case.

One of the main results of \cite{matsumoto2010reverse} is an axiomatic characterization of the Umegaki relative entropy \cite{umegaki1962conditional}. In an analogous way, our results lead to a new characterization of the sandwiched R\'enyi quasi-entropies of order $\alpha\ge 1$: suppose that a quantity $f(\rho,\sigma)$, depending on pairs of positive semidefinite operators $\rho,\sigma$ on arbitrary finite dimensional Hilbert spaces, satisfies the following properties:
\begin{enumerate}[(i)]
\item $f(\rho_1\otimes\rho_2,\sigma_1\otimes\sigma_2)=f(\rho_1,\sigma_1)f(\rho_2,\sigma_2)$ (multiplicativity)
\item $f(\rho_1\oplus\rho_2,\sigma_1\oplus\sigma_2)=f(\rho_1,\sigma_1)+f(\rho_2,\sigma_2)$ (additivity)
\item $f(I_n,I_n)=n$ (normalization)
\item $f(T(\rho),T(\sigma))\le f(\rho,\sigma)$ when $T$ is a completely positive trace-nonincreasing map (data processing inequality)
\item $f$ is increasing in the first and decreasing in the second argument (with respect to the semidefinite partial order). (monotonicity)
\end{enumerate}
Then $f(\rho,\sigma)=\sandwichedquasientropy{\alpha}{\rho}{\sigma}=\Tr\left(\sigma^{\frac{1-\alpha}{2\alpha}}\rho\sigma^{\frac{1-\alpha}{2\alpha}}\right)^\alpha$ for some $\alpha\ge 1$.

We point out that our work does not rely on existing proofs of the data processing inequality for the sandwiched R\'enyi divergences of order $\alpha>1$, but rather provides a new one. Since the components of this proof are somewhat scattered around, we briefly summarize the high-level structure for readers who wish to focus on this aspect of our work. First we show that, when restricted to classical (commuting) pairs, $\thesandwichedquasientropy{\alpha}$ with $\alpha\ge 1$ are monotone homomorphisms (Proposition~\ref{prop:cmonotones}). Then we show that, still restricting to classical pairs, there are no other monotone homomorphisms (Proposition~\ref{prop:cmultiplicative} and Remark~\ref{rem:relativesubmajorization}). We show that, informally, quantum dichotomies are bounded from below and above by classical dichotomies, which by general considerations implies that every monotone homomorphism on classical pairs has at least one (monotone, homomorphic) extension to quantum pairs (Corollary~\ref{cor:surjective}). Finally, with the help of the pinching inequality we show that the restriction of $\thesandwichedquasientropy{\alpha}$ to classical pairs has at most one extension to quantum pairs, namely $\thesandwichedquasientropy{\alpha}$ (Thereom~\ref{thm:qmonotones}).

The remainder of this paper is structured as follows. In Section~\ref{sec:preliminaries} we collect definitions and facts related to preordered semirings and the pinching maps to be used in later sections. In Section~\ref{sec:dichotomysemiring} we define the preordered semiring of dichotomies and prove that it satisfies the technical conditions required by the generalization of Strassen's theorem. In Section~\ref{sec:spectralpoints} we complete the classification of real-valued monotone semiring homomorphisms. Section~\ref{sec:rates} translates our results on the asymptotic spectrum to various settings, in the context of pair transformations, hypothesis testing and quantum thermodynamics.

\section{Preliminaries}\label{sec:preliminaries}

\subsection{Preordered semirings}

In this section we collect definitions and results related to preordered semirings, asymptotic preorders and asymptotic spectra. A (commutative) semiring is a set $S$ equipped with two binary operations $+$, $\cdot$ that are both commutative and associative, such that neutral elements $0$ and $1$ exist, $0a=0$ and $(a+b)c=ac+bc$ for all $a,b,c\in S$. A preordered semiring is a commutative semiring with a preorder $\preorderle$ such that $a\preorderle b$ implies $a+c\preorderle b+c$ and $ac\preorderle bc$ for all $a,b,c\in S$ and such that $0\preorderle 1$. Note in particular that $0\preorderle a$ for all $a\in S$. We say that $S$ is of \emph{polynomial growth} \cite{fritz2018generalization} if there is an element $u$ such that for every nonzero $x\in S$ there is a $k\in\mathbb{N}$ such that $x\preorderle u^k$ and $1\preorderle u^kx$. Any such $u$ is called \emph{power universal}.

From now on $u$ will denote an arbitrary but fixed power universal element. For our purposes the precise choice does not matter. In particular, in the following definition we define an asymptotic relaxation of the preorder \cite[Definition 2.3]{vrana2020generalization}. While the definition involves the element $u$, the resulting relation would be the same if we chose a different one. In the context of ordered commutative monoids, the relaxed preorder is closely related to regularized rates \cite[8.16. Definition]{fritz2017resource}.
\begin{definition}[asymptotic preorder]\label{def:asymptoticpreorder}
Let $a,b\in S$. We say that $a$ is asymptotically less than $b$ (notation: $a\asymptoticle b$) if there exists a sequence $(k_n)_{n\in\mathbb{N}}$ of nonnegative integers such that
\begin{align}
\lim_{n\to\infty}\frac{k_n}{n} & = 0 \label{eq:asymptoticpreorderdeflimit}
\intertext{(i.e. $k_n$ is sublinear) and}
\forall n\in\mathbb{N}:a^n\preorderle u^{k_n}b^n. \label{eq:asymptoticpreorderdefinequality}
\end{align}
\end{definition}

If $S_1$ and $S_2$ are preordered semirings, then a map $\varphi:S_1\to S_2$ is a monotone semiring homomorphism if it satisfies $\varphi(0)=0$, $\varphi(1)=1$, $\varphi(a+b)=\varphi(a)+\varphi(b)$, $\varphi(ab)=\varphi(a)\varphi(b)$ and is compatible with the preorders, i.e. $a\preorderle_1 b\implies\varphi(a)\preorderle_2(b)$ for all $a,b\in S_1$. The asymptotic spectrum of a preordered semiring $S$ is the set of monotone semiring homomorphisms into the semiring of nonnegative reals (with its usual addition, multiplication and order), i.e. maps $f:S\to\mathbb{R}_{\ge 0}$ satisfying for all $x,y\in S$
\begin{enumerate}[(i)]
\item $f(x+y)=f(x)+f(y)$
\item $f(xy)=f(x)f(y)$
\item $x\preorderle y\implies f(x)\le f(y)$
\item $f(1)=1$.
\end{enumerate}
We denote the asymptotic spectrum by $\Delta(S,\preorderle)$. The elements of the asymptotic spectrum are called spectral points. For every $s\in S$ the evaluation map $\ev_s:\Delta(S,\preorderle)\to\mathbb{R}_{\ge 0}$ is defined as $\ev_s(f)=f(s)$.

It is clear that for any $f\in\Delta(S,\preorderle)$ and $a,b\in S$ such that $a\asymptoticle b$ we have $f(a)\le f(b)$. This follows by applying $f$ to the inequalities \eqref{eq:asymptoticpreorderdefinequality}, taking roots and the limit $n\to\infty$ using \eqref{eq:asymptoticpreorderdeflimit}. The converse in general does not hold but it turns out that under additional hypotheses it does:
\begin{theorem}[{\cite[Theorem 1.2.]{vrana2020generalization}}]\label{thm:spectrumpreorder}
Let $(S,\preorderle)$ be a preordered semiring of polynomial growth such that the canonical map $\mathbb{N}\hookrightarrow S$ is an order embedding, and let $M\subseteq S$ and $S_0$ the subsemiring generated by $M$. Suppose that
\begin{enumerate}[({M}1)]
\item\label{it:invertibleuptobounded} for all $s\in S\setminus\{0\}$ there exist $m\in M$ and $n\in\mathbb{N}$ such that $1\preorderle nms$ and $ms\preorderle n$
\item\label{it:boundedev} for all $m\in M$ such that $\ev_m:\Delta(S_0)\to\mathbb{R}_{\ge 0}$ is bounded there is an $n\in\mathbb{N}$ such that $m\preorderle n$.
\end{enumerate}
Then for every $x,y\in S$ we have
\begin{equation}
x\asymptoticge y\iff\forall f\in\Delta(S,\preorderle):f(x)\ge f(y).
\end{equation}
\end{theorem}

It is often convenient to express the asymptotic comparison of a pair of elements in terms of rates (note that this is the maximal regularized rate in the sense of \cite[8.21. Remark]{fritz2017resource}), defined as
\begin{equation}
R(x\to y)=\sup\setbuild{r\in\mathbb{R}_{\ge 0}}{\exists(k_n)_{n\in\mathbb{N}}\text{ sublinear:}\forall n\in\mathbb{N}:u^{k_n}x^n\preorderge y^{\lceil r n\rceil}}.
\end{equation}
$x\asymptoticge y$ is equivalent to $R(x\to y)\ge 1$ and in general $R(x\to y)$ can be understood as a way to measure the relative strength of the resources represented by $x$ and $y$. Under the conditions of Theorem~\ref{thm:spectrumpreorder}, the rate is given by
\begin{equation}\label{eq:generalrateformula}
R(x\to y)=\sup\setbuild{r\in\mathbb{R}_{\ge 0}}{\forall f\in\Delta(S,\preorderle):\log f(x)\ge r\log f(y)}.
\end{equation}
Assuming $x,y\ge 1$, the supremum is equal to
\begin{equation}\label{eq:specialrateformula}
R(x\to y)=\inf_{\substack{f\in\Delta(S,\preorderle)  \\  f(y)\neq 1}}\frac{\log f(x)}{\log f(y)}.
\end{equation}

A monotone semiring homomorphism $\varphi:S_1\to S_2$ induces a map
\begin{equation}
\Delta(\varphi):\Delta(S_2,\preorderle_2)\to\Delta(S_1,\preorderle_1)
\end{equation}
between the asymptotic spectra, which sends $f$ to $f\circ\varphi$. We will be interested in the special case when the homomorphism is the inclusion of a subsemiring that contains ``large'' elements so that the following proposition applies.
\begin{proposition}[{\cite[Proposition 3.9.]{vrana2020generalization}}]\label{prop:surjective}
Let $(S,\preorderle)$ be a preordered semiring of polynomial growth and $S_0\le S$ a subsemiring such that $\forall s\in S\setminus\{0\}\exists r,q\in S_0$ such that $1\preorderle rs\preorderle q$. Let $i:S_0\hookrightarrow S$ be the inclusion. Then $\Delta(i)$ is surjective.
\end{proposition}

\subsection{Pinching}

One of the main technical tools relating classical and quantum dichotomies is the pinching map. Let $A$ be a normal operator on a finite dimensional Hilbert space $\mathcal{H}$ and let
\begin{equation}
A=\sum_{\lambda\in\spectrum(A)}\lambda P_{\lambda}
\end{equation}
be its spectral decomposition where the operators $P_{\lambda}$ are pairwise orthogonal projections. The $A$-pinching map $\pinching{A}:\boundeds(\mathcal{H})\to\boundeds(\mathcal{H})$ is defined as
\begin{equation}
\pinching{A}(X)=\sum_{\lambda\in\spectrum(A)}P_{\lambda}XP_{\lambda}.
\end{equation}
The most important properties of the pinching map are summarized in the following lemma.
\begin{lemma}[see e.g. \cite{dye1952radon} and \cite{dixmier1981neumann}]\label{lem:pinchingproperties}
Let $A$ be as above.
\begin{enumerate}[(i)]
\item\label{it:pinchingcpt} $\pinching{A}$ is completely positive and trace preserving.
\item\label{it:pinchingfix} $\pinching{A}(A)=A$.
\item\label{it:pinchingcommutes} $\pinching{A}(X)$ commutes with $A$ for every $X$.
\item\label{it:pinchinginequality} If $X\ge 0$ then $X\le|\spectrum(A)|\mathcal{P}_A(X)$. (pinching inequality)
\item\label{it:pinchingrandomunitary} $\pinching{A}$ is a convex combination of unitary conjugations.
\end{enumerate}
\end{lemma}

\begin{proposition}\label{prop:pinchingpowerlimit}
Let $\alpha\ge 1$, $A,X\in\boundeds(\mathcal{H})$, $A$ normal and $X\ge 0$. Then
\begin{equation}
\lim_{n\to\infty}\sqrt[n]{\Tr\left(\pinching{A^{\otimes n}}(X^{\otimes n})\right)^\alpha}=\Tr X^\alpha.
\end{equation}
\end{proposition}
\begin{proof}
We make use of \ref{it:pinchinginequality} and \ref{it:pinchingrandomunitary} of Lemma~\ref{lem:pinchingproperties}:
\begin{equation}
|\spectrum(A^{\otimes n})|^{-1}X^{\otimes n}\le\pinching{A^{\otimes n}}(X^{\otimes n})=\sum_{i\in I}p_iU_i X^{\otimes n}U_i^*,
\end{equation}
where $(p_i)_{i\in I}$ is a probability vector and $U_i$ are unitaries on $\mathcal{H}$. The function $M\mapsto\Tr M^\alpha$ is monotone increasing and convex on the set of positive semidefinite operators, therefore
\begin{equation}
\begin{split}
|\spectrum(A^{\otimes n})|^{-\alpha}(\Tr X^\alpha)^n
 & = \Tr(|\spectrum(A^{\otimes n})|^{-1}X^{\otimes n})^\alpha  \\
 & \le \Tr\left(\pinching{A^{\otimes n}}(X^{\otimes n})\right)^\alpha  \\
 & = \Tr\left(\sum_{i\in I}p_iU_i X^{\otimes n}U_i^*\right)^\alpha  \\
 & \le \sum_{i\in I}p_i\Tr\left(U_i X^{\otimes n}U_i^*\right)^\alpha
  = (\Tr X^\alpha)^n.
\end{split}
\end{equation}
The claim follows by taking the $n$th root and limits $n\to\infty$, noting that $1\le|\spectrum(A^{\otimes n})|\le(n+1)^{\dim\mathcal{H}}$, a consequence of \cite[Lemma 2.2]{csiszar2011information}.
\end{proof}

\section{Dichotomies and pair transformations}\label{sec:dichotomysemiring}

In this section we describe the preordered semiring of quantum dichotomies and the associated asymptotic preorder, and then show that it satisfies the conditions of Theorem~\ref{thm:spectrumpreorder} and therefore the asymptotic preorder is characterized by monotone semiring homomorphisms. It will be convenient to work with the following relaxed notion of a quantum dichotomy.
\begin{definition}[dichotomy]
A quantum dichotomy on a finite dimensional Hilbert space $\mathcal{H}$ is a pair $(\rho,\sigma)$ where $\rho,\sigma\in\boundeds(\mathcal{H})$ are positive semidefinite, $\support\rho\subseteq\support\sigma$, and if $\rho=0$ then $\sigma=0$. A classical dichotomy is a quantum dichotomy $(\rho,\sigma)$ which in addition satisfies $\rho\sigma=\sigma\rho$.

A dichotomy $(\rho,\sigma)$ is called normalized if $\Tr\rho=\Tr\sigma=1$.
\end{definition}
We call two dichotomies $(\rho,\sigma)$ and $(\rho',\sigma')$ on $\mathcal{H}$ and $\mathcal{H}'$ equivalent if there is a partial isometry $U:\mathcal{H}\to\mathcal{H}'$ such that $U\rho U^*=\rho'$, $U\sigma U^*=\sigma'$, $U^*\rho' U=\rho$ and $U^*\sigma' U=\sigma$. In other words, equivalence means that the pairs are essentially the same, up to possibly enlarging one of the Hilbert spaces so that they have the same size, followed by a unitary rotation.

Technically, we do not wish to distinguish equivalent dichotomies, but work instead with equivalence classes. It is clear that every dichotomy is equivalent to one on $\mathbb{C}^d$ for some $d\in\mathbb{N}$, therefore we can form the set of equivalence classes by taking the quotient of the set of dichotomies on $\mathbb{C}^d$ for all $d$ by this equivalence relation. Nevertheless, we will frequently gloss over this distinction to avoid cumbersome wording, and pretend that dichotomies themselves are the elements.

A classical dichotomy may equivalently be thought of as a pair of measures $(p,q)$ on a finite set $\mathcal{X}$. We will think of these as collections of the nonnegative real numbers $(p_x)_{x\in\mathcal{X}},(q_x)_{x\in\mathcal{X}}$ that the measures associate with the one-element subsets.

We let $\dichotomies$ denote the set of equivalence classes of quantum dichotomies and $\cdichotomies$ the subset of equivalence classes of classical dichotomies. Our next goal is to equip both sets with binary operations that turn them into commutative semirings.
\begin{definition}[addition, multiplication of dichotomies]
Let $(\rho,\sigma)$ and $(\rho',\sigma')$ be dichotomies on the Hilbert spaces $\mathcal{H}$ and $\mathcal{H}'$, respectively. The sum $(\rho,\sigma)+(\rho',\sigma')$ is the dichotomy $(\rho\oplus\rho',\sigma\oplus\sigma')$ on the Hilbert space $\mathcal{H}\oplus\mathcal{H}'$. The product $(\rho,\sigma)(\rho',\sigma')$ is the dichotomy $(\rho\otimes\rho',\sigma\otimes\sigma')$ on the Hilbert space $\mathcal{H}\otimes\mathcal{H}'$.
\end{definition}
Both the addition and the multiplication induce well-defined operations on $\dichotomies$, which are easily seen to be commutative and associative, and satisfy the distributive law. The equivalence class of the dichotomy $(0,0)$ on $\mathbb{C}$ is the neutral element for addition, whereas the equivalence class of the dichotomy $(1,1)$ on $\mathbb{C}$ is the neutral element for multiplication. Therefore $\dichotomies$ is a semiring. It is clear that $\cdichotomies\subseteq\dichotomies$ is a subsemiring.

The final ingredient is a preorder that generalizes relative majorization to allow comparison of unnormalized states, and in particular induces the usual ordering on the natural numbers, represented by pairs $(I,I)$ on $\mathbb{C}^n$.
\begin{definition}\label{def:preorderdef}
Let $(\rho,\sigma)$ and $(\rho',\sigma')$ be dichotomies on the Hilbert spaces $\mathcal{H}$ and $\mathcal{H}'$, respectively. We say that the pair $(\rho,\sigma)$ is greater than $(\rho',\sigma')$ and write $(\rho,\sigma)\preorderge(\rho',\sigma')$ if there exists a completely positive trace-nonincreasing map $T:\boundeds(\mathcal{H})\to\boundeds(\mathcal{H}')$ such that
\begin{subequations}
\begin{align}
T(\rho) & \ge \rho'  \label{eq:dichotomypreorderdefinequality} \\
T(\sigma) & = \sigma'.  \label{eq:dichotomypreorderdefequality}
\end{align}
\end{subequations}
\end{definition}
This gives a well-defined relation $\preorderle$ on $\dichotomies$ which is clearly reflexive and transitive.

The operational interpretation of this preorder is the following. Let $(\rho,\sigma)$ and $(\rho',\sigma')$ be normalized dichotomies and $a,b\in(0,1]$. The relation $(\rho,\sigma)\preorderge(a\rho',b\sigma')$ means that there is an instrument with a distinguished outcome (``success'') corresponding to $T$, with the following properties: when applied to $\sigma$, the probability of the successful outcome is $b$ and in this case the post-measurement state is $\sigma'$; when applied to $\rho$, the probability of success satisfies $p_s=\Tr T(\rho)\ge a$ and the post-measurement state is bounded from below by $\frac{a}{p_s}\rho'$. The latter condition implies that the post-measurement state is close to $\rho'$:
\begin{equation}
\begin{split}
\frac{1}{2}\norm[1]{\frac{T(\rho)}{\Tr T(\rho)}-\rho'}
 & = \frac{1}{2}\norm[1]{\frac{T(\rho)-a\rho'}{\Tr T(\rho)}-\left(1-\frac{a}{\Tr T(\rho)}\right)\rho'}  \\
 & \le \frac{1}{2}\norm[1]{\frac{T(\rho)-a\rho'}{\Tr T(\rho)}}+\frac{1}{2}\norm[1]{\left(1-\frac{a}{\Tr T(\rho)}\right)\rho'}  \\
 & = \frac{1}{2}\frac{\Tr T(\rho)-a}{\Tr T(\rho)}+\frac{1}{2}\left(1-\frac{a}{\Tr T(\rho)}\right)  \\
 & = 1-\frac{a}{\Tr T(\rho)}  \\
 & \le 1-a.
\end{split}
\end{equation}

We remark that our definition is closely related but not identical to the notion of relative submajorization, introduced in \cite{renes2016relative}. In the definition of the latter, the condition \eqref{eq:dichotomypreorderdefequality} is relaxed to $T(\sigma)\le\sigma'$. The subsequent analysis works equally well for relative submajorization, with minor changes (see Remark~\ref{rem:relativesubmajorization} below).

\begin{proposition}\label{prop:dichotomiespreorderedsemiring}
$(\dichotomies,\preorderle)$ is a preordered semiring.
\end{proposition}
\begin{proof}
Recall that $(0,0)$ and $(1,1)$ represent the additive and multiplicative neutral elements. Choosing $T=0$ we see that $(0,0)\preorderge(1,1)$.

We need to verify the compatibility of the relation with the binary operations. Let $(\rho,\sigma)$, $(\rho',\sigma')$ and $(\omega,\tau)$ be dichotomies on the Hilbert spaces $\mathcal{H},\mathcal{H}'$ and $\mathcal{K}$, and suppose that $(\rho',\sigma')\preorderle(\rho,\sigma)$. This means that there exists a completely positive trace-nonincreasing map $T:\boundeds(\mathcal{H})\to\boundeds(\mathcal{H}')$ such that \eqref{eq:dichotomypreorderdefinequality} and \eqref{eq:dichotomypreorderdefequality} hold. Let $\tilde{T}:\boundeds(\mathcal{H}\oplus\mathcal{K})\to\boundeds(\mathcal{H}'\oplus\mathcal{K})$ be the linear map given by
\begin{equation}
\tilde{T}\left(\begin{bmatrix}
A & B  \\
C & D
\end{bmatrix}
\right)=\begin{bmatrix}
T(A) & 0  \\
0 & D
\end{bmatrix},
\end{equation}
where the block structures correspond to the direct sum decompositions above. Then $\tilde{T}$ is completely positive and trace nonincreasing and satisfies
\begin{subequations}
\begin{align}
\tilde{T}(\rho\oplus\omega) & = T(\rho)\oplus\omega \ge \rho'\oplus\omega  \\
\tilde{T}(\sigma\oplus\omega) & = T(\sigma)\oplus\omega = \sigma'\oplus\omega,
\end{align}
\end{subequations}
therefore $(\rho',\sigma')+(\omega,\tau)\preorderle(\rho,\sigma)+(\omega,\tau)$.

Similarly, the map $T\otimes\id_{\boundeds(\mathcal{K})}$ satisfies
\begin{subequations}
\begin{align}
(T\otimes\id_{\boundeds(\mathcal{K})})(\rho\otimes\omega) & = T(\rho)\otimes\omega \ge \rho'\otimes\omega  \\
(T\otimes\id_{\boundeds(\mathcal{K})})(\sigma\otimes\omega) & = T(\sigma)\otimes\omega = \sigma'\otimes\omega,
\end{align}
\end{subequations}
which shows that $(\rho',\sigma')(\omega,\tau)\preorderle(\rho,\sigma)(\omega,\tau)$.
\end{proof}

\begin{proposition}\label{prop:polynomialgrowth}
$(\dichotomies,\preorderle)$ is of polynomial growth. More precisely, the dichotomy $(3,2)$ on $\mathbb{C}$ is power universal.
\end{proposition}
\begin{proof}
Let $(\rho,\sigma)$ be a dichotomy on $\mathcal{H}$ and $u=(3,2)$. By choosing $T=\frac{1}{2}\id$ in Definition~\ref{def:preorderdef} we can see that $u\preorderge 1$. Let
\begin{equation}
k_1=\max\left\{0,\lceil\log\Tr\sigma\rceil,\left\lceil\frac{\log\norm[\infty]{\sigma^{-1/2}\rho\sigma^{-1/2}}}{\log(3/2)}\right\rceil\right\}.
\end{equation}
Then $\Tr(2^{-k_1}\sigma)\le 1$, therefore $T_1(x)=x2^{-k_1}\sigma$ defines a completely positive trace nonincreasing map $T_1:\boundeds(\mathbb{C})\to\boundeds(\mathcal{H})$. It satisfies
\begin{subequations}
\begin{align}
T_1(3^{k_1}) & = (3/2)^{k_1}\sigma\ge\norm[\infty]{\sigma^{-1/2}\rho\sigma^{-1/2}}\sigma\ge\rho  \\
T_1(2^{k_1}) & = \sigma,
\end{align}
\end{subequations}
therefore $u^{k_1}=(3^{k_1},2^{k_1})\preorderge(\rho,\sigma)$.

Let
\begin{equation}
k_2=\max\left\{0,\lceil-\log\Tr\sigma\rceil,\left\lceil\frac{\log\frac{\Tr\sigma}{\Tr\rho}}{\log(3/2)}\right\rceil\right\}.
\end{equation}
Then $2^{-k_2}(\Tr\sigma)\le 1$, therefore $T_2(x)=2^{-k_2}\frac{\Tr x}{\Tr\sigma}$ defines a completely positive trace nonincreasing map $T_2:\boundeds(\mathcal{H})\to\boundeds(\mathbb{C})$. It satisfies
\begin{subequations}
\begin{align}
T_2(3^{k_2}\rho) & = (3/2)^{k_2}\frac{\Tr\rho}{\Tr\sigma}\ge 1  \\
T_2(2^{k_2}\sigma) & = 2^{-k_2}\frac{\Tr 2^{k_2}\sigma}{\Tr\sigma}=1,
\end{align}
\end{subequations}
therefore $u^{k_2}(\rho,\sigma)=(3^{k_2}\rho,2^{k_2}\sigma)\preorderge 1$.

With $k=\max\{k_1,k_2\}$ both $u^k(\rho,\sigma)\preorderge 1$ and $u^k\preorderge(\rho,\sigma)$ hold.
\end{proof}

Since $\dichotomies$ is of polynomial growth, we can form the asymptotic preorder as in Definition~\ref{def:asymptoticpreorder}. To see how the definition specializes to our semiring, let $(\rho,\sigma)$ and $(\rho',\sigma')$ be normalized dichotomies and $R,r\ge0$. The relation $(\rho,\sigma)\asymptoticge(2^{-R}\rho',2^{-r}\sigma')$ means that for every $n$ there is a two-outcome instrument with one outcome associated with $T_n$ and interpreted as success, with the following properties:
\begin{enumerate}[(i)]
\item when $T_n$ is applied to $\sigma^{\otimes n}$, the probability of success is $2^{-nr+o(n)}$ and upon observing this outcome the post-measurement state is ${\sigma'}^{\otimes n}$
\item when applied to $\rho^{\otimes n}$, the probability of success is $p_n=2^{-nR+o(n)}$ and upon observing this outcome the post-measurement state is bounded from below by $\frac{2^{-nR+o(n)}}{p_n}{\sigma'}^{\otimes n}$.
\end{enumerate}
In the second case the condition implies that the post-measurement state has distance at most $1-2^{-nR+o(n)}$ from ${\sigma'}^{\otimes n}$, but it is stronger than merely requiring this estimate: ${\sigma'}^{\otimes n}$ may have eigenvalues that are smaller than $2^{-nR+o(n)}$, but even these cannot be completely suppressed (in particular the support of the post-measurement state must contain that of ${\sigma'}^{\otimes n}$).

Our next goal is to verify that the conditions of Theorem~\ref{thm:spectrumpreorder} hold for $(\dichotomies,\preorderle)$. The role of $M$ will be played by the dichotomies on $\mathbb{C}$. Recall that condition \ref{it:invertibleuptobounded} is that every nonzero element of $\dichotomies$ can be multiplied with a suitable element of $M$ in such a way that the product is bounded from below and above by natural numbers. First we present a sufficient condition for this boundedness property (which is also necessary, but we do not need this).
\begin{proposition}\label{prop:boundeddichotomy}
Let $(\rho,\sigma)$ be a dichotomy on $\mathcal{H}$, and suppose that $\rho\le\sigma$ and $\rank(\sigma-\rho)<\rank\sigma$. Then there is an $n\in\mathbb{N}$ such that $1\preorderle n(\rho,\sigma)$ and $(\rho,\sigma)\preorderle n$.
\end{proposition}
\begin{proof}
Let $P$ be a nonzero projection such that $P(\sigma-\rho)=0$. Let $n_1=\lceil(\Tr\sigma P)^{-1}\rceil$. Then $T_1(x)=(n_1\Tr(\sigma P))^{-1}\Tr(x(I\otimes P))$ defines a completely positive trace nonincreasing map $T_1:\boundeds(\mathbb{C}^{n_1}\otimes\mathcal{H})\to\boundeds(\mathbb{C})$. It satisfies
\begin{subequations}
\begin{align}
T_1(I_{n_1}\otimes\rho) & = (n_1\Tr(\sigma P))^{-1}\Tr(I_{n_1}\otimes\rho(I\otimes P))=1  \\
T_1(I_{n_1}\otimes\sigma) & = (n_1\Tr(\sigma P))^{-1}\Tr(I_{n_1}\otimes\sigma(I\otimes P))=1,
\end{align}
\end{subequations}
therefore $n_1(\rho,\sigma)\preorderge 1$.

Let $n_2=\lceil\Tr\sigma\rceil$. Then $T_2(x)=\frac{\sigma}{n_2}(\Tr x)$ defines a completely positive trace nonincreasing map $T_2:\boundeds(\mathbb{C}^{n_2})\to\boundeds(\mathcal{H})$. It satisfies
\begin{equation}
T_2(I)=\frac{\sigma}{n_2}(\Tr I)=\sigma\ge\rho,
\end{equation}
therefore $n_2\preorderge(\rho,\sigma)$.

With $n=\max\{n_1,n_2\}$ both $n(\rho,\sigma)\preorderge 1$ and $n\preorderge(\rho,\sigma)$ hold.
\end{proof}

Condition~\ref{it:boundedev} requires an inner approximation of the asymptotic spectrum of the subsemiring generated by $M$. Since $M$ consists of classical dichotomies, so does this subsemiring. For this reason we will exhibit a family of elements of $\Delta(\cdichotomies,\preorderle)$ (later we will see that the given set of elements is almost exhaustive, but this is not necessary for verifying condition~\ref{it:boundedev}). Note that a spectral point $f\in\Delta(\dichotomies,\preorderle)$ can be evaluated on elements of $\dichotomies$ (and similarly for $\cdichotomies$), which are equivalence classes of pairs. To simplify notation, we will effectively regard $f$ as a function of two variables and write e.g. $f(\rho,\sigma)$ for its value on the equivalence class of the dichotomy $(\rho,\sigma)$.
\begin{proposition}\label{prop:cmonotones}
Let $\alpha\ge1$ and consider the map $f_\alpha:\cdichotomies\to\mathbb{R}_{\ge0}$ defined as
\begin{equation}
f_\alpha(p,q)=\sum_{x\in\mathcal{X}}p_x^\alpha q_x^{1-\alpha}.
\end{equation}
Then $f_\alpha\in\Delta(\cdichotomies,\preorderle)$.
\end{proposition}
\begin{proof}
It is clear that $f_\alpha$ is a semiring homomorphism for every $\alpha\ge 1$, thus we need to show that it is also monotone. This immediately follows from the data processing inequality for the R\'enyi divergence since $f_\alpha(p,q)=2^{(\alpha-1)\relativeentropy[\alpha]{p}{q}}$, but for completeness we provide a direct proof.

Suppose that $(p,q)$ and $(r,s)$ are classical dichotomies characterized by the probabilities $p_x,q_x,r_y,s_y$ ($x\in\mathcal{X}$, $y\in\mathcal{Y}$) and $(p,q)\preorderge(r,s)$. The ordering means that there is a completely positive trace-nonincreasing map $T:\boundeds(\mathbb{C}^{\mathcal{X}})\to\boundeds(\mathbb{C}^{\mathcal{X}'})$ such that
\begin{align}\label{eq:pqgreaterrs}
\sum_{x\in\mathcal{X}}p_x T(\ketbra{x}{x}) & \ge\sum_{y\in\mathcal{Y}}r_y\ketbra{y}{y}  \\
\sum_{x\in\mathcal{X}}q_x T(\ketbra{x}{x}) & =\sum_{y\in\mathcal{Y}}s_y\ketbra{y}{y}.
\end{align}
The dephasing maps $D_{\mathcal{X}}:\boundeds(\mathbb{C}^{\mathcal{X}})\to\boundeds(\mathbb{C}^{\mathcal{X}})$ given by $D_{\mathcal{X}}(\ketbra{x_1}{x_2})=\delta_{x_1,x_2}\ketbra{x_1}{x_2}$ are completely positive and $D_{\mathcal{X}}$ ($D_{\mathcal{Y}}$) fixes $p$ and $q$ ($r$ and $s$), therefore we can replace $T$ with $D_{\mathcal{Y}}\circ T\circ D_{\mathcal{X}}$. This composition is essentially a classical substochastic map with matrix entries $T_{xy}:=\bra{y}T(\ketbra{x}{x})\ket{y}$. Let $Q_y=\sum_{x\in\mathcal{X}}T_{xy}q_x$. Then
\begin{equation}
\begin{split}
f_\alpha(r,s)
 & = \sum_{y\in\support Q}r_y^\alpha s_y^{1-\alpha}  \\
 & \le \sum_{y\in\support Q}\left(\sum_{x\in\support q}T_{xy}p_x\right)^\alpha\left(\sum_{x\in\support q}T_{xy}q_x\right)^{1-\alpha}  \\
 & = \sum_{y\in\support Q}Q_y\left(\sum_{x\in\support q}\frac{T_{xy}q_x}{Q_y}\frac{p_x}{q_x}\right)^\alpha  \\
 & \le \sum_{y\in\support Q}Q_y\sum_{x\in\support q}\frac{T_{xy}q_x}{Q_y}\left(\frac{p_x}{q_x}\right)^\alpha  \\
 & = \sum_{x\in\support q}q_x\left(\frac{p_x}{q_x}\right)^\alpha\sum_{y\in\support Q}T_{xy}  \\
 & \le \sum_{x\in\support q}q_x\left(\frac{p_x}{q_x}\right)^\alpha=f_\alpha(p,q),
\end{split}
\end{equation}
where the first inequality uses \eqref{eq:pqgreaterrs}, the second inequality uses convexity of $x\mapsto x^\alpha$ and the third inequality uses that $\sum_yT_{xy}\le 1$.
\end{proof}

We are now in a position to prove that $\dichotomies$ satisfies the conditions of Theorem~\ref{thm:spectrumpreorder}. The first condition is verified with an application of Proposition~\ref{prop:boundeddichotomy}, while for the second condition we use the spectral points presented in Proposition~\ref{prop:cmonotones} in the large $\alpha$ limit.
\begin{proposition}\label{prop:multiplierset}
Let $M$ be the set of dichotomies on $\mathbb{C}$ and $S_0\subseteq\dichotomies$ the subsemiring generated by $M$.
\begin{enumerate}[(i)]
\item\label{it:orderembedding} The map $\mathbb{N}\hookrightarrow\dichotomies$ is an order embedding.
\item\label{it:Mmultiplier} For every dichotomy $(\rho,\sigma)\neq 0$ there is an $m\in M$ and $n\in\mathbb{N}$ such that $1\preorderle m(\rho,\sigma)\preorderle n$.
\item\label{it:Mspectrumbounded} If $m\in M$ and $\ev_m:\Delta(S_0,\preorderle)\to\mathbb{R}_{\ge 0}$ is bounded, then there is an $n\in\mathbb{N}$ such that $m\preorderle n$.
\end{enumerate}
\end{proposition}
\begin{proof}
\ref{it:orderembedding}: Let $n,n'\in\mathbb{N}$. The corresponding elements in $\dichotomies$ are represented by $(I_n,I_n)$ and $(I_{n'},I_{n'})$, where $I_n$ is the identity on $\mathbb{C}^n$. Since $(0,0)\preorderle(1,1)$, if $n\le n'$ then also $(I_n,I_n)\preorderle(I_{n'},I_{n'})$. On the other hand, if $n>n'$ then a completely positive trace nonincreasing map $T:\boundeds(\mathbb{C}^{n'})\to\boundeds(\mathbb{C}^{n})$ cannot satisfy $T(I_{n'})\ge I_n$ since this would require $n'=\Tr I_{n'}\ge T(I_{n'})\ge n>n'$.

\ref{it:Mmultiplier}: Let $(\rho,\sigma)\neq 0$ be a dichotomy and let
\begin{equation}
\lambda=\min\setbuild{t\in\mathbb{R}}{\rho\le t\sigma}=\norm[\infty]{\sigma^{-1/2}\rho\sigma^{-1/2}},
\end{equation}
where the inverse is understood on the support of $\sigma$. Then $(1,\lambda)(\rho,\sigma)=(\rho,\lambda\sigma)$ satisfies the conditions of Proposition~\ref{prop:boundeddichotomy}, therefore we can choose $m=(1,\lambda)$.

\ref{it:Mspectrumbounded}: Let $p,q\in\mathbb{R}_{>0}$ and consider the element $(p,q)\in M$. Since $S_0\subseteq\cdichotomies$, every function in Proposition~\ref{prop:cmonotones} also gives rise to an element of $\Delta(S_0,\preorderle)$ by restriction. Suppose that $\ev_m$ is bounded. Then
\begin{equation}
\infty>\limsup_{\alpha\to\infty}f_\alpha(p,q)=q\lim_{\alpha\to\infty}\left(\frac{p}{q}\right)^\alpha,
\end{equation}
which is equivalent to $p\le q$. By the argument in the proof of Proposition~\ref{prop:boundeddichotomy} we see that $(p,q)\preorderle n$ for some $n\in\mathbb{N}$ (namely $n=\lceil q\rceil$).
\end{proof}

An important byproduct of the findings above is that Proposition~\ref{prop:surjective} applies to the inclusion of $\cdichotomies$ in $\dichotomies$, which means that every monotone homomorphism $f:\cdichotomies\to\mathbb{R}_{\ge0}$ has at least one (monotone, homomorphic) extenstion to $\dichotomies$.
\begin{corollary}\label{cor:surjective}
Let $i:\cdichotomies\hookrightarrow\dichotomies$ be the inclusion. Then $\Delta(i)$ is surjective.
\end{corollary}
\begin{proof}
Let $(\rho,\sigma)$ be a dichotomy. By \ref{it:Mmultiplier} of Proposition~\ref{prop:multiplierset} there is an $m\in M$ and $1\preorderle m(\rho,\sigma)\preorderle n$. Both $m$ and $n=(I_n,I_n)$ are classical, therefore the condition in Proposition~\ref{prop:surjective} is satisfied by the subsemiring $\cdichotomies$. We conclude that $\Delta(i)$ is surjective.
\end{proof}

\section{Spectral points}\label{sec:spectralpoints}

In this section we describe all the elements of $\Delta(\dichotomies,\preorderle)$ explicitly. To this end, we first find every element of $\Delta(\cdichotomies,\preorderle)$ and show that each of them has at most one extension to $\dichotomies$. Together with Corollary~\ref{cor:surjective} this implies that there is exactly one extension.
\begin{proposition}\label{prop:cmultiplicative}
Let $f\in\Delta(\cdichotomies,\preorderle)$. Then there is an $\alpha\in\{0\}\cup[1,\infty)$ such that $f(p,q)=p^\alpha q^{1-\alpha}$ for a dichotomy $(p,q)$ on $\mathbb{C}$ (with the convention $0^\alpha 0^{1-\alpha}=0$).
\end{proposition}
\begin{proof}
Let $g(x)=f(x,x)$ for $x\in\mathbb{R}_{\ge 0}$ and $h(y)=f(y,1)$ for $y\in\mathbb{R}_{>0}$. By multiplicativity of $f$, both $g$ and $h$ satisfy the Cauchy functional equation $g(x_1x_2)=g(x_1)g(x_2)$ and $h(y_1y_2)=h(y_1)h(y_2)$ and in addition
\begin{equation}\label{eq:fnumbers}
f(p,q)=f(q,q)f\left(\left(\frac{p}{q},1\right)\right)=g(q)h\left(\frac{p}{q}\right).
\end{equation}

If $0\le x_1\le x_2$, then choosing $T=\frac{x_1}{x_2}\id_{\boundeds(\mathbb{C})}$ in Definition~\ref{def:preorderdef} shows that $(x_1,x_1)\preorderle(x_2,x_2)$, thus $g$ must be monotone increasing. Therefore $g(x)=x^\beta$ for some $\beta\ge 0$.

If $0<y_1\le y_2$, then choosing $T=\id_{\boundeds(\mathbb{C})}$ in Definition~\ref{def:preorderdef} shows that $(y_1,1)\preorderle(y_2,1)$, therefore $h$ is monotone increasing. This implies $h(y)=y^\alpha$ for some $\alpha\ge 0$. By \eqref{eq:fnumbers} we have
\begin{equation}
f(p,q)=q^\beta\left(\frac{p}{q}\right)^\alpha=p^\alpha q^{\beta-\alpha}
\end{equation}

Consider the dichotomies $(1,1)+(1,1)$ on $\mathbb{C}^2$ and $(2,2)$ on $\mathbb{C}$. Then we have both $(1,1)+(1,1)\preorderle(2,2)$ and $(1,1)+(1,1)\preorderge(2,2)$, as can be seen by choosing $T=\Tr$ and the map
\begin{equation}
T(1)=\frac{1}{2}\begin{bmatrix}
1 & 0  \\
0 & 1
\end{bmatrix},
\end{equation}
respectively. This implies
\begin{equation}
2=f((1,1)+(1,1))=f(2,2)=g(2)=2^\beta,
\end{equation}
i.e. $\beta=1$.

Finally, we have $(2,1)+(1,2)\preorderge(3,3)$ in the semiring (choose $T=\Tr$), therefore
\begin{equation}
2^\alpha+2^{1-\alpha}\ge 3,
\end{equation}
which implies $\alpha\notin(0,1)$.
\end{proof}
\begin{remark}\label{rem:relativesubmajorization}
If we used relative submajorization as the preorder, then $(1,1)$ would be greater than $(1,2)$, therefore monotonicity would require $1=f(1,1)\ge f(1,2)=2^{1-\alpha}$, i.e. $1\le\alpha$.
\end{remark}

Now we can complete the classification of elements in the asymptotic spectrum of classical dichotomies.
\begin{theorem}\label{thm:cmonotones}
$\Delta(\cdichotomies,\preorderle)=\setbuild{f_\alpha}{\alpha\in\{0\}\cup[1,\infty)}$, where the functions $f_\alpha:\cdichotomies\to\mathbb{R}_{\ge 0}$ are given by
\begin{equation}\label{eq:cmonotone}
f_\alpha(p,q)=\sum_{x\in\mathcal{X}}p_x^\alpha q_x^{1-\alpha},
\end{equation}
with the convention $0^0=1$.
\end{theorem}
\begin{proof}
We saw in Proposition~\ref{prop:cmonotones} that $f_\alpha\in\Delta(\cdichotomies,\preorderle)$ for $\alpha\ge 1$. The map $f_0(p,q)=\Tr q$ is clearly additive, multiplicative, normalized and monotone under trace nonincreasing maps. We show that these monotones exhaust the set of monotone homomorphisms allowed by Proposition~\ref{prop:cmultiplicative}.

Let $f\in\Delta(\cdichotomies,\preorderle)$. The semiring $S_0$ is additively generated by dichotomies on $\mathbb{C}$. More precisely, elements can be represented by pairs $(p,q)$ where
\begin{align}\label{eq:pqrepresentative}
p & = \sum_{x\in\mathcal{X}}p_x\ketbra{x}{x}  \\
q & = \sum_{x\in\mathcal{X}}q_x\ketbra{x}{x}
\end{align}
for some finite set $\mathcal{X}$ and strictly positive numbers $(p_x)_{x\in\mathcal{X}},(q_x)_{x\in\mathcal{X}}$, and with the addition in the semiring we may write this as
\begin{equation}
(p,q)=\sum_{x\in\mathcal{X}}(p_x,q_x).
\end{equation}
This means that a monotone semiring homomorphism $f:S_0\to\mathbb{R}_{\ge 0}$ is in fact uniquely determined by its values on dichotomies on $\mathbb{C}$. Thus the restriction of $f$ to $S_0$ must agree with $f_\alpha$ for some $\alpha$ in the indicated range.

Similarly, elements of $\cdichotomies$ are also of the form \eqref{eq:pqrepresentative} with $0\le p_x<q_x$ allowed as long as $p\neq 0$. For every $\epsilon>0$ the inequality $(p+\epsilon q,q)\preorderge(p,q)$ holds (choose $T=\id$), therefore
\begin{equation}\label{eq:specialupperbound}
f(p,q)\le\liminf_{\epsilon\to 0+}f(p+\epsilon q,q)
\end{equation}
Next consider the map
\begin{equation}
T(x)=\epsilon\frac{q}{\norm[1]{q}}\Tr(x)+(1-\epsilon)x,
\end{equation}
where $\epsilon\in[0,1]$. $T$ is (completely) positive and trace preserving, therefore by Definition~\ref{def:preorderdef} we have
\begin{equation}
(p,q)\preorderge(T(p),T(q))=\left(\epsilon\frac{q}{\norm[1]{q}}\Tr(p)+(1-\epsilon)p,q\right).
\end{equation}
$f$ is monotone, therefore
\begin{equation}\label{eq:speciallowerbound}
f(p,q)\ge\limsup_{\epsilon\to0+}f\left(\epsilon\frac{q}{\norm[1]{q}}\Tr(p)+(1-\epsilon)p,q\right).
\end{equation}
The right hand sides of \eqref{eq:specialupperbound} and \eqref{eq:speciallowerbound} involve only elements of $S_0$, therefore we can evaluate $f$ on them as $f_\alpha$, which leads to
\begin{equation}
\begin{split}
f(p,q)
 & \le \liminf_{\epsilon\to 0+}f_\alpha(p+\epsilon q,q)  \\
 & = \lim_{\epsilon\to 0+}\sum_{x\in\support q}(p_x+\epsilon q_x)^\alpha q_x^{1-\alpha}  \\
 & = \sum_{x\in\support q}p_x^\alpha q_x^{1-\alpha}
\end{split}
\end{equation}
and
\begin{equation}
\begin{split}
f(p,q)
 & \ge \limsup_{\epsilon\to0+}f_\alpha\left(\epsilon\frac{q}{\norm{q}}\Tr(p)+(1-\epsilon)p,q\right)  \\
 & \ge \lim_{\epsilon\to0+}\sum_{x\in\support q}\left(\epsilon\frac{q_x}{\norm[1]{q}}\Tr(p)+(1-\epsilon)p_x\right)^\alpha q_x^{1-\alpha}  \\
 & = \sum_{x\in\support q}p_x^\alpha q_x^{1-\alpha},
\end{split}
\end{equation}
in both cases with the convention $0^0:=\lim_{\epsilon\to0+}\epsilon^0=1$.
\end{proof}

In the following theorem we reduce the classification of monotone semiring homomorphisms on quantum dichotomies to the classical ones with the help of the pinching map.
\begin{theorem}\label{thm:qmonotones}
Consider the functions $f_\alpha:\dichotomies\to\mathbb{R}_{\ge 0}$ given by
\begin{equation}\label{eq:qmonotone}
\tilde{f}_\alpha(\rho,\sigma)=\begin{cases}
\sandwichedquasientropy{\alpha}{\rho}{\sigma}=\Tr\left(\sigma^{\frac{1-\alpha}{2\alpha}}\rho\sigma^{\frac{1-\alpha}{2\alpha}}\right)^\alpha & \text{if $\alpha\ge 1$}  \\
\Tr\sigma & \text{if $\alpha=0$.}
\end{cases}
\end{equation}
Then $\Delta(\dichotomies,\preorderle)=\setbuild{\tilde{f}_\alpha}{\alpha\in\{0\}\cup[1,\infty)}$.
\end{theorem}
\begin{proof}
If $\tilde{f}\in\Delta(\dichotomies,\preorderle)$, then its restriction to $\cdichotomies$ is an element of $\Delta(\cdichotomies,\preorderle)$. These are of the form $f_\alpha$ for $\alpha\in\{0\}\cup[1,\infty)$ (Theorem~\ref{thm:cmonotones}). Conversely, by Corollary~\ref{cor:surjective} the function $f_\alpha$ has at least one extension $\tilde{f}\in\Delta(\dichotomies,\preorderle)$. We show that that there is at most one extension for every possible $\alpha$ and that it agrees with $\tilde{f}_\alpha$.

Suppose that $\tilde{f}\in\Delta(\dichotomies,\preorderle)$ is an extension of $f_\alpha$ with $\alpha\in\{0\}\cup[1,\infty)$, and let $(\rho,\sigma)$ be a dichotomy. For every $n\in\mathbb{N}$ the following inequalities hold in $\dichotomies$:
\begin{equation}
(\rho^{\otimes n},\sigma^{\otimes n})\preorderge(\pinching{\sigma^{\otimes n}}(\rho^{\otimes n}),\sigma^{\otimes n})\preorderge(|\spectrum(\sigma^{\otimes n})|^{-1}\rho^{\otimes n},\sigma^{\otimes n}).
\end{equation}
The first one follows by applying the pinching map $\pinching{\sigma^{\otimes n}}$ to both parts of the dichotomy and using that it fixes $\sigma^{\otimes n}$ (see \ref{it:pinchingcpt} and \ref{it:pinchingfix} of Lemma~\ref{lem:pinchingproperties}). In the second step we decrease the first element, as can be seen from \ref{it:pinchinginequality} of Lemma~\ref{lem:pinchingproperties}.

Since $\tilde{f}$ is multiplicative and monotone, we have
\begin{equation}\label{eq:qcestimate}
\begin{split}
\tilde{f}(\rho,\sigma)
 & = \left(\tilde{f}(\rho^{\otimes n},\sigma^{\otimes n})\right)^{1/n}  \\
 & \ge \left(\tilde{f}(\pinching{\sigma^{\otimes n}}(\rho^{\otimes n}),\sigma^{\otimes n})\right)^{1/n}  \\
 & \ge \left(\tilde{f}(|\spectrum(\sigma^{\otimes n})|^{-1}\rho^{\otimes n},\sigma^{\otimes n})\right)^{1/n}  \\
 & = \left(\tilde{f}(|\spectrum(\sigma^{\otimes n})|^{-1},1)\right)^{1/n}\tilde{f}(\rho,\sigma).
\end{split}
\end{equation}
The dichotomies $(\pinching{\sigma^{\otimes n}}(\rho^{\otimes n}),\sigma^{\otimes n})$ and $(|\spectrum(\sigma^{\otimes n})|^{-1},1)$ are classical (\ref{it:pinchingcommutes} of Lemma~\ref{lem:pinchingproperties}), therefore $\tilde{f}$ and $f_\alpha$ agree on them. Since $1\le|\spectrum(\sigma^{\otimes n})|\le(n+1)^{\dim\mathcal{H}}$, we have
\begin{equation}\label{eq:spectrumsizelimit}
\begin{split}
\lim_{n\to\infty}\left(\tilde{f}(|\spectrum(\sigma^{\otimes n})|^{-1},1)\right)^{1/n}
 & = \lim_{n\to\infty}\left(f_\alpha(|\spectrum(\sigma^{\otimes n})|^{-1},1)\right)^{1/n}  \\
 & = \lim_{n\to\infty}|\spectrum(\sigma^{\otimes n})|^{-\alpha/n}=1.
\end{split}
\end{equation}

For $\alpha=0$ we get from \eqref{eq:qcestimate} with $n=1$ the chain of inequalities $\tilde{f}(\rho,\sigma)\ge\tilde{f}(\pinching{\sigma}(\rho),\sigma)\ge\tilde{f}(\rho,\sigma)$, therefore $\tilde{f}(\rho,\sigma)=f_0(\pinching{\sigma}(\rho),\sigma)=\Tr\sigma$.

For $\alpha\ge 1$, \eqref{eq:qcestimate} and \eqref{eq:spectrumsizelimit} imply
\begin{equation}
\begin{split}
\tilde{f}(\rho,\sigma)
 & = \lim_{n\to\infty}\left(\tilde{f}(\pinching{\sigma^{\otimes n}}(\rho^{\otimes n}),\sigma^{\otimes n})\right)^{1/n}  \\
 & = \lim_{n\to\infty}\left(\Tr\pinching{\sigma^{\otimes n}}(\rho^{\otimes n})^\alpha(\sigma^{\otimes n})^{1-\alpha}\right)^{1/n}  \\
 & = \lim_{n\to\infty}\left(\Tr\left((\sigma^{\otimes n})^{\frac{1-\alpha}{2\alpha}}\pinching{\sigma^{\otimes n}}(\rho^{\otimes n})(\sigma^{\otimes n})^{\frac{1-\alpha}{2\alpha}}\right)^{\alpha}\right)^{1/n}  \\
 & = \lim_{n\to\infty}\left(\Tr\left(\pinching{\sigma^{\otimes n}}\left((\sigma^{\otimes n})^{\frac{1-\alpha}{2\alpha}}\rho^{\otimes n}(\sigma^{\otimes n})^{\frac{1-\alpha}{2\alpha}}\right)\right)^{\alpha}\right)^{1/n}  \\
 & = \lim_{n\to\infty}\left(\Tr\left(\pinching{\sigma^{\otimes n}}\left((\sigma^{\frac{1-\alpha}{2\alpha}}\rho\sigma^{\frac{1-\alpha}{2\alpha}})^{\otimes n}\right)\right)^{\alpha}\right)^{1/n}  \\
 & = \Tr\left(\sigma^{\frac{1-\alpha}{2\alpha}}\rho\sigma^{\frac{1-\alpha}{2\alpha}}\right)^{\alpha}.
\end{split}
\end{equation}
The second and third equalities use that $\pinching{\sigma^{\otimes n}}(\rho^{\otimes n})$ and $\sigma^{\otimes n}$ commute, the fourth one uses that $\sigma^{\frac{1-\alpha}{2\alpha}}$ commutes with the projections appearing in the pinching map, and the last step uses Proposition~\ref{prop:pinchingpowerlimit}.
\end{proof}
For normalized dichotomimes $(\rho,\sigma)$ the function $f_\alpha$ in \eqref{eq:qmonotone} is the sandwiched quasi-entropy \cite{wilde2014strong} and can be expressed in terms of the sandwiched R\'enyi divergence \cite{muller2013quantum,wilde2014strong} as
\begin{equation}\label{eq:falphasandwiched}
f_\alpha(\rho,\sigma)=2^{(\alpha-1)\sandwicheddivergence{\alpha}{\rho}{\sigma}}.
\end{equation}
Note that when $\Tr\rho\neq 1$, it is customary to include a $-\frac{1}{\alpha-1}\log\Tr\rho$ term in the definition of the sandwiched R\'enyi divergence \cite{muller2013quantum}.

\section{Rates for probabilistic asymptotic pair transformations}\label{sec:rates}

Now that we have an explicit description of asymptotic spectrum of the semiring of dichotomies, we specialize the rate formula \eqref{eq:generalrateformula}. First let us consider normalized dichotomies $(\rho,\sigma)$ and $(\rho',\sigma')$ on $\mathcal{H}$, $\mathcal{H}'$. Then $(\rho,\sigma)\preorderge 1$ and $(\rho',\sigma')\preorderge 1$ (choose $T=\Tr$ in Definition~\ref{def:preorderdef}), therefore \eqref{eq:specialrateformula} can be used. Noting that $f_0(\rho,\sigma)=1=f_0(\rho',\sigma')$, the rate is given by the same formula for both of the preorders:
\begin{equation}\label{eq:subexponentialrateformula}
R((\rho,\sigma)\to(\rho',\sigma'))=\inf_{\alpha>1}\frac{\log f_\alpha(\rho,\sigma)}{\log f_\alpha(\rho',\sigma')}=\inf_{\alpha>1}\frac{\sandwicheddivergence{\alpha}{\rho}{\sigma}}{\sandwicheddivergence{\alpha}{\rho'}{\sigma'}}.
\end{equation}
This means that there are substochastic maps $T_n:\boundeds(\mathcal{H}^{\otimes n})\to\boundeds({\mathcal{H}'}^{\otimes R((\rho,\sigma)\to(\rho',\sigma'))n+o(n)})$ that map $n$ copies of $\sigma$ to roughly $R((\rho,\sigma)\to(\rho',\sigma'))n$ copies of $\sigma'$ with a probability that decays slower than any exponential, and at the same time maps $n$ copies of $\rho$ to a subnormalized state $\tau$ satisfying $\tau\ge 2^{-o(n)}{\sigma'}^{\otimes R((\rho,\sigma)\to(\rho',\sigma'))n}$. The latter condition in turn implies that the probability of failure as well as the approximation error is $1-2^{-o(n)}$. However, our notion of asymptotic transformations is more restrictive than merely requiring this error since the inequality \eqref{eq:dichotomypreorderdefinequality} implies that even exponentially small components cannot be suppressed completely.

Our rate formula should be contrasted to that of \cite{buscemi2019information} and \cite{wang2019resource}, i.e. the ratio of (Umegaki) relative entropies, which is valid for approximate transformations in the first component with asymptotically vanishing error. The results of \cite{buscemi2019information} show that above the rate $\relativeentropy{\rho}{\sigma}/\relativeentropy{\rho'}{\sigma'}$ the approximation error goes to $1$ exponentially. This is consistent with our rate being in general lower since at rate \eqref{eq:subexponentialrateformula} we can guarantee an error that approaches $1$ slower than any exponential.

In \cite{wang2019resource} dichotomies have been interpreted in terms of a resource theory of asymmetric distinguishability. When considering approximate pair transformations with asymptotically vanishing error, the resource theory is shown to be reversible. The dichotomy $(\ketbra{0}{0},\pi)$ with $\pi=\frac{1}{2}(\ketbra{0}{0}+\ketbra{1}{1})$ serves as a possible unit, referred to as one bit of asymmetric distinguishability, and $\relativeentropy{\rho}{\sigma}$ is the asymptotic value of the dichotomy $(\rho,\sigma)$.

Let us calculate the distillation and dilution rates in the sense of our asymptotic transformations. The spectral points in $\Delta(\dichotomies,\preorderle)$ evaluate to
\begin{equation}
f_\alpha(\ketbra{0}{0},\pi)=\begin{cases}
2^{\alpha-1} & \text{if $\alpha\ge 1$}  \\
1 & \text{if $\alpha=0$},
\end{cases}
\end{equation}
therefore $\sandwicheddivergence{\alpha}{\ketbra{0}{0}}{\pi}=1$. If we wish to distill from $(\rho,\sigma)$ bits of asymmetric distinguishability, then the optimal rate using \eqref{eq:subexponentialrateformula} is
\begin{equation}
R((\rho,\sigma)\to(\ketbra{0}{0},\pi))=\inf_{\alpha>1}\sandwicheddivergence{\alpha}{\rho}{\sigma}=\relativeentropy{\rho}{\sigma},
\end{equation}
which is equal to the rate when an asymptotically vanishing error is allowed. For the reverse task of diluting bits of asymmetric distinguishability to $(\rho,\sigma)$, the rate becomes
\begin{equation}
R((\rho,\sigma)\to(\ketbra{0}{0},\pi))=\inf_{\alpha>1}\frac{1}{\sandwicheddivergence{\alpha}{\rho}{\sigma}}=\frac{1}{\sandwicheddivergence{\infty}{\rho}{\sigma}},
\end{equation}
where $\sandwicheddivergence{\infty}{\rho}{\sigma}=\log\norm[\infty]{\sigma^{-1/2}\rho\sigma^{-1/2}}$ is the max-divergence \cite{renner2008security,datta2009min}. This rate is equal to the inverse of the asymptotic cost in the exact dilution task \cite{wang2019resource}.

We note that the form and the meaning of our rate formula \eqref{eq:subexponentialrateformula} is reminiscent of the zero-exponent limit of the rate formula for bipartite entanglement transformations \cite{jensen2019asymptotic}. In both cases the value corresponds to the rate where the theory of asymptotic spectra guarantees a success probability that decays slower than any exponenital in the number of copies. For entanglement transformations an independent calculation shows that with the same rate it is actually possible to have a success probability going to $1$ \cite{jensen2019asymptotic}, therefore we expect that a similar improvement is possible in the present case as well.

\subsection{Strong converse exponents}

We turn to the asymptotic comparison of possibly unnormalized dichotomies. Recall that $M$ is the set of nonzero dichotomies on $\mathbb{C}$. Every nonzero dichotomy is the product of a normalized one with an element of $M$ in a unique way. Since $M$ consist of invertible elements, we may assume without loss of generality that the initial pair is normalized.

Let $(\rho,\sigma)$ and $(\rho',\sigma')$ be normalized dichotomies and $r,R\ge 0$. The inequality $(\rho,\sigma)\asymptoticge(2^{-R}\rho',2^{-r}\sigma')$ means that there exists a sequence of trace-nonincreasing channels $T_n:\boundeds(\mathcal{H}^{\otimes n})\to\boundeds({\mathcal{H}'}^{\otimes n})$ that transform $\sigma^{\otimes n}$ to ${\sigma'}^{\otimes n}$ exactly with probability $2^{-rn+o(n)}$, and that transform $\rho^{\otimes n}$ to a subnormalized state that is larger than $2^{-Rn+o(n)}{\rho'}^{\otimes n}$. As before, the latter implies a success probability at least $2^{-Rn+o(n)}$ and approximation error at most $1-2^{-Rn+o(n)}$. The optimal exponent is
\begin{equation}\label{eq:converseerrorexponent}
\begin{split}
R^*(r)
 & = \inf\setbuild{R\ge 0}{(\rho,\sigma)\asymptoticge(2^{-R},2^{-r})(\rho',\sigma')}  \\
 & = \inf\setbuild{R\ge 0}{\forall f\in\Delta(\dichotomies,\preorderge):f(\rho,\sigma)\ge f(2^{-R},2^{-r})f(\rho',\sigma')}  \\
 & = \inf\setbuild{R\ge 0}{\forall \alpha>1:\log f_\alpha(\rho,\sigma)\ge -\alpha R-(1-\alpha)r+\log f_\alpha(\rho',\sigma')}  \\
 & = \inf\setbuild{R\ge 0}{\forall \alpha>1:R\ge-\frac{1-\alpha}{\alpha}r-\frac{1}{\alpha}\log f_\alpha(\rho,\sigma)+\frac{1}{\alpha}\log f_\alpha(\rho',\sigma')}  \\
 & = \sup_{\alpha>1}\frac{\alpha-1}{\alpha}\left[r-\sandwicheddivergence{\alpha}{\rho}{\sigma}+\sandwicheddivergence{\alpha}{\rho'}{\sigma'}\right].
\end{split}
\end{equation}
The second equality is an application of Theorem~\ref{thm:spectrumpreorder}, the third equality uses Theorem~\ref{thm:qmonotones} and that $f_0(\rho,\sigma)=f_1(\rho,\sigma)=1$ and $f_0(2^{-R},2^{-r})=2^{-r}\le 1$ and $f_1(2^{-R},2^{-r})=2^{-R}\le 1$, while the last equality uses \eqref{eq:falphasandwiched}.

More generally, we may wish to transform $n$ copies of $(\rho,\sigma)$ to $\kappa n+o(n)$ copies of $(\rho',\sigma')$ (for some $\kappa>0$) with the same probabilities -- in this case $f_\alpha(\rho',\sigma')$ needs to be replaced with $f_\alpha(\rho',\sigma')^\kappa$ in the condition above. This captures the trade-off between the exponents $r,R$ and the rate $\kappa$. The optimal value of any of these can be expressed in terms of the other two:
\begin{equation}
\begin{split}
R^*(\kappa,r)
 & = \inf\setbuild{R\ge0}{\forall\alpha>1:\log f_\alpha(\rho,\sigma)\ge -\alpha R-(1-\alpha)r+\kappa\log f_\alpha(\rho',\sigma')}  \\
 & = \inf\setbuild{R\ge0}{\forall\alpha>1:R\ge -\frac{1-\alpha}{\alpha}r-\frac{1}{\alpha}\log f_\alpha(\rho,\sigma)+\frac{\kappa}{\alpha}\log f_\alpha(\rho',\sigma')}  \\
 & = \sup_{\alpha>1}\frac{\alpha-1}{\alpha}\left[r-\sandwicheddivergence{\alpha}{\rho}{\sigma}+\kappa\sandwicheddivergence{\alpha}{\rho'}{\sigma'}\right].
\end{split}
\end{equation}
\begin{equation}
\begin{split}
\kappa^*(r,R)
 & = \sup\setbuild{\kappa}{\forall\alpha>1:\log f_\alpha(\rho,\sigma)\ge -\alpha R-(1-\alpha)r+\kappa\log f_\alpha(\rho',\sigma')}  \\
 & = \sup\setbuild{\kappa}{\forall\alpha>1:\alpha R+(1-\alpha)r+\log f_\alpha(\rho,\sigma)\ge \kappa\log f_\alpha(\rho',\sigma')}  \\
 & = \inf_{\alpha>1}\frac{\alpha R+(1-\alpha)r+\log f_\alpha(\rho,\sigma)}{\log f_\alpha(\rho',\sigma')}  \\
 & = \inf_{\alpha>1}\frac{\frac{\alpha}{\alpha-1} R-r+\sandwicheddivergence{\alpha}{\rho}{\sigma}}{\sandwicheddivergence{\alpha}{\rho'}{\sigma'}}
\end{split}
\end{equation}
\begin{equation}
\begin{split}
r^*(\kappa,R)
 & = \sup\setbuild{R\ge0}{\forall\alpha>1:\log f_\alpha(\rho,\sigma)\ge -\alpha R-(1-\alpha)r+\kappa\log f_\alpha(\rho',\sigma')}  \\
 & = \sup\setbuild{R\ge0}{\forall\alpha>1:\alpha R+\log f_\alpha(\rho,\sigma)-\kappa\log f_\alpha(\rho',\sigma')\ge (\alpha-1)r}  \\
 & = \inf_{\alpha>1}\frac{\alpha}{\alpha-1} R+\sandwicheddivergence{\alpha}{\rho}{\sigma}-\kappa\sandwicheddivergence{\alpha}{\rho'}{\sigma'}
\end{split}
\end{equation}

\subsection{Hypothesis testing}

We consider the hypothesis testing problem of distinguishing two sources of independent and identically distributed copies of one of two quantum states $\rho$ and $\sigma$ on $\mathcal{H}$. The obeserver is allowed to perform a measurement on $n$ copies, described by a two-outcome POVM $(\Pi_n,I-\Pi_n)$ where $\Pi_n\in\boundeds(\mathcal{H}^n)$ satisfies $0\le\Pi_n\le I$. The outcome associated with $\Pi_n$ results in accepting the null hypothesis $\rho$, whereas alternative hypothesis gets accepted upon obtaining the other outcome.

A Type I error occurs when the state in question was $\rho$ but the observer finds that it was $\sigma$. This happens with probability $\alpha_n(\Pi_n):=\Tr(\rho^{\otimes n}(I-\Pi_n))$. A Type II error occurs in the opposite case, when the state was $\sigma$ but the null hypothesis is accepted. This happens with probability $\beta_n(\Pi_n)=\Tr(\sigma^{\otimes n}\Pi_n)$. For every $n$ there is a trade-off between the probabilities of the two kinds of errors and we are interested in the possible behaviors of both probabilities in the limit $n\to\infty$.

The quantum Stein's lemma \cite{hiai1991proper} says that for all $\epsilon\in(0,1)$, for every sequence of measurements $(\Pi_n)_{n\in\mathbb{N}}$ under the condition $\alpha_n(\Pi_n)\le\epsilon$ the Type II error probability satisfies
\begin{equation}\label{eq:Steinlemma}
\limsup_{n\to\infty}-\frac{1}{n}\log\beta_n(\Pi_n)\ge\relativeentropy{\rho}{\sigma},
\end{equation}
and there is a sequence of measurements attaining this value as a limit. The strong converse property, proved in \cite{ogawa2005strong}, states that for any sequence that does not satisfy \eqref{eq:Steinlemma} the probability of the Type I error necessarily converges to $1$ exponentially fast. The smallest possible exponent, called the strong converse exponent was found in \cite{mosonyi2015quantum}. We now show how one can obtain the same result with our methods.

First note that there is a bijection between two-outcome measurements (tests) on a Hilbert space $\mathcal{H}$ and completely positive trace-nonincreasing maps $T:\boundeds(\mathcal{H})\to\boundeds(\mathbb{C})\simeq\mathbb{C}$. Indeed, given a POVM $(\Pi,I-\Pi)$ on $\mathcal{H}$ we can define the map $T(x)=\Tr(x\Pi)$ and conversely, any linear map $T:\boundeds(\mathcal{H})\to\mathbb{C}$ is of the form $T(x)=\Tr(xA)$ for some $A\in\boundeds(\mathcal{H})$, and such a map is (completely) positive iff $A\ge 0$ and trace-nonincreasing iff $A\le I$. For this reason, when $(\rho,\sigma)$ is a normalized dichotomy on $\mathcal{H}$ and $(a,b)$ is a dichotomy on $\mathbb{C}$, we may restate the condition for $(\rho,\sigma)\preorderge(a,b)$ as $\exists\Pi\in\boundeds(\mathcal{H})$, $0\le\Pi\le I$ such that
\begin{align}
1-\alpha(\Pi)=\Tr\rho\Pi & \ge a  \\
\beta(\Pi)=\Tr\sigma\Pi & = b.
\end{align}
Accordingly, the asymptotic ordering $(\rho,\sigma)\asymptoticge(2^{-R},2^{-r})$ means that there is a sequence of measurement operators $(\Pi_n)_{n\in\mathbb{N}}$ such that
\begin{align}
1-\alpha_n(\Pi_n) & \ge 2^{-Rn+o(n)}  \\
\beta_n(\Pi_n) & = 2^{-rn+o(n)}.
\end{align}
We obtain the strong converse exponent as a function of $r\ge0$ by specializing \eqref{eq:converseerrorexponent} to the target pair $(1,1)$:
\begin{equation}
R^*(r) = \sup_{\alpha>1}\frac{\alpha-1}{\alpha}\left[r-\sandwicheddivergence{\alpha}{\rho}{\sigma}\right].
\end{equation}

\subsection{Work-assisted transformations}

In the resource theory approach to thermodynamics, one fixes a background inverse temperature $\beta$ and models thermal operations as energy-preserving unitaries acting jointly on the system in question and an arbitrary heat bath at inverse temperature $\beta$. Gibbs-preserving maps are a convenient relaxation of thermal operations, and form the free operations in the resource theory of athermality. Classically, the possible state transformations are identical for these two choices of allowed operations \cite{horodecki2013fundamental}, whereas in the quantum setting Gibbs-preserving maps are strictly more powerful than thermal operations \cite{faist2015gibbs}.

Consider a quantum system with Hilbert space $\mathcal{H}$ and Hamiltonian $H\in\boundeds(\mathcal{H})$. The Gibbs state at temperature $\beta$ is $\gamma_{H,\beta}=\frac{1}{Z(\beta)}2^{-\beta H}$ where the normalizing factor $Z(\beta)=\Tr 2^{-\beta H}$ is the partition function (note that in statistical mechanics the base of logarithms and exponentials is usually $e$ -- all the formulas below remain valid with this choice). Note that a linear map preserves $\gamma_{H,\beta}$ iff it preserves $2^{-\beta H}$, therefore we will omit the normalizing factor.

Let $\mathcal{H}_B=\mathbb{C}^2$ be a second Hilbert space modelling a battery with Hamiltonian
\begin{equation}
H_B(w)=\begin{bmatrix}
0 & 0  \\
0 & w
\end{bmatrix}=w\ketbra{1}{1}.
\end{equation}
Changing the state of the battery from $\ketbra{1}{1}$ to $\ketbra{0}{0}$ can be interpreted as drawing work $w$ from it. It may happen that the transformation $\rho\to\rho'$ becomes thermodynamically possible if at the same time we draw some amount $w$ of work from a battery, where the two-component system is described by the Hamiltonian $H\otimes\id_{\mathbb{C}^2}+\id_{\mathcal{H}}\otimes H_B(w)$. The smallest such $w$ is the work cost of the transformation.

Asymptotically, if we allow joint Gibbs-preserving maps that succeed with a probability that decays slower than any exponential, then the condition for a transformation is
\begin{equation}
f_\alpha(\rho\otimes\ketbra{1}{1},2^{-\beta H}\otimes 2^{-\beta H_B(w)})\ge f_\alpha(\rho'\otimes\ketbra{0}{0},2^{-\beta H}\otimes 2^{-\beta H_B(w)})
\end{equation}
for all $\alpha\in\{0\}\cup[1,\infty)$. $f_0$ and $f_1$ are the traces of the second and first component of the dichotomy, which evaluate to the same value on both sides. Using multiplicativity of $f_\alpha$ we can write the asymptotic work cost per copy as
\begin{equation}\label{eq:workcost}
\begin{split}
w^*
 & = \inf\setbuild{w\in\mathbb{R}}{\forall\alpha>1:2^{-(1-\alpha)\beta w}f_\alpha(\rho,2^{-\beta H})\ge f_\alpha(\rho',2^{-\beta H})}  \\
 & = \inf\setbuild{w\in\mathbb{R}}{\forall\alpha>1:(\alpha-1)\beta w+\log f_\alpha(\rho,2^{-\beta H})\ge \log f_\alpha(\rho',2^{-\beta H})}  \\
 & = \sup_{\alpha>1}\frac{1}{(\alpha-1)\beta}\log f_\alpha(\rho',2^{-\beta H})-\frac{1}{(\alpha-1)\beta}\log f_\alpha(\rho,2^{-\beta H})  \\
 & = \sup_{\alpha>1}\frac{1}{\beta}\sandwicheddivergence{\alpha}{\rho'}{2^{-\beta H}}-\frac{1}{\beta}\sandwicheddivergence{\alpha}{\rho}{2^{-\beta H}}.
\end{split}
\end{equation}

We note that
\begin{equation}
\frac{1}{\beta}\relativeentropy{\rho}{2^{-\beta H}}=\frac{1}{\beta}\Tr\rho(\log\rho-\log 2^{-\beta H})=\frac{1}{\beta}\Tr\rho(\beta H+\log\rho)=E-\frac{1}{\beta}\entropy(\rho)
\end{equation}
is the Helmholtz free energy, therefore $\frac{1}{\beta}\sandwicheddivergence{\alpha}{\rho}{2^{-\beta H}}$ may be thought of as a free energy of order $\alpha$ \cite{brandao2015second}. When the state $\rho$ is the Gibbs state $\frac{1}{Z(\beta)}2^{-\beta H}$, the R\'enyi divergences become independent of $\alpha$ and are equal to $-\log Z(\beta)$. Since in general $\alpha\mapsto\sandwicheddivergence{\alpha}{\rho}{2^{-\beta H}}$ is an increasing function, we can evaluate \eqref{eq:workcost} if either the initial or the final state is the Gibbs state. If $\rho=\gamma_{H,\beta}$ the work cost \eqref{eq:workcost} becomes
\begin{equation}
w^*=\frac{1}{\beta}\sandwicheddivergence{\infty}{\rho'}{2^{-\beta H}}+\frac{1}{\beta}\log Z(\beta)=\frac{1}{\beta}\sandwicheddivergence{\infty}{\rho'}{\gamma_{H,\beta}}.
\end{equation}
In contrast, if the final state is $\rho'=\gamma_{H,\beta}$, \eqref{eq:workcost} simplifies to
\begin{equation}
w^*=-\frac{1}{\beta}\log Z(\beta)-\frac{1}{\beta}\relativeentropy{\rho}{2^{-\beta H}}=-\frac{1}{\beta}\relativeentropy{\rho}{\gamma_{H,\beta}},
\end{equation}
where the negative sign indicates that work is extracted in the process.

More generally, suppose that the transformation is probabilistic and approximate, while it still preserves the Gibbs state exactly. If the success probability is allowed to decay as $2^{-Rn}$ for some $R\ge0$, then the relevant inequality is $(\rho\otimes\ketbra{1}{1},2^{-\beta H}\otimes 2^{-\beta H_B(w)})\asymptoticge(2^{-R}\rho'\otimes\ketbra{0}{0},2^{-\beta H}\otimes 2^{-\beta H_B(w)})$. Assuming that we invest work $w$ per copy, the slowest decay rate is given by (here a negative supremum should be understood as $R^*(w)=0$)
\begin{equation}
\begin{split}
R^*(w)
 & = \inf\setbuild{R\in\mathbb{R}_{\ge 0}}{\forall\alpha>1:2^{-(1-\alpha)\beta w}f_\alpha(\rho,2^{-\beta H})\ge 2^{-\alpha R}f_\alpha(\rho',2^{-\beta H})}  \\
 & = \inf\setbuild{R\in\mathbb{R}_{\ge 0}}{\forall\alpha>1:(\alpha-1)\beta w+\log f_\alpha(\rho,2^{-\beta H})\ge -\alpha R+\log f_\alpha(\rho',2^{-\beta H})}  \\
 & = \inf\setbuild{R\in\mathbb{R}_{\ge 0}}{\forall\alpha>1:\alpha R\ge \log f_\alpha(\rho',2^{-\beta H})-\log f_\alpha(\rho,2^{-\beta H})-(\alpha-1)\beta w}  \\
 & = \sup_{\alpha>1}\frac{1}{\alpha}\log f_\alpha(\rho',2^{-\beta H})-\frac{1}{\alpha}\log f_\alpha(\rho,2^{-\beta H})-\frac{\alpha-1}{\alpha}\beta w  \\
 & = \beta\sup_{\alpha>1}\frac{\alpha-1}{\alpha}\left[\frac{1}{\beta}\sandwicheddivergence{\alpha}{\rho'}{2^{-\beta H}}-\frac{1}{\beta}\sandwicheddivergence{\alpha}{\rho}{2^{-\beta H}}-w\right].
\end{split}
\end{equation}
In the special case where the initial state is the Gibbs state, this expression simplifies to
\begin{equation}
R^*(w) = \max\left\{0,\sandwicheddivergence{\infty}{\rho'}{\gamma_{H,\beta}}-\beta w\right\}.
\end{equation}

\section*{Acknowledgement}

We thank Alexander M\"uller-Hermes for useful discussions. This work was supported by the \'UNKP-19-4 New National Excellence Program of the Ministry for Innovation and Technology and the Bolyai J\'anos Research Fellowship of the Hungarian Academy of Sciences. P.~Vrana acknowledges support from the Hungarian National Research, Development and Innovation Office (NKFIH) within the Quantum Technology National Excellence Program (Project Nr.~2017-1.2.1-NKP-2017-00001) and via the research grants K124152, KH129601.
A.~H.~Werner thanks the VILLUM FONDEN for its support with a Villum Young Investigator Grant (Grant No.~25452) and its support via the QMATH Centre of Excellence (Grant No.~10059).

\bibliography{refs}{}

\end{document}